\documentclass[journal]{IEEEtran}
\ifCLASSINFOpdf
\else
\fi

\usepackage{amsmath}
\usepackage{amsthm}
\usepackage{graphicx}
\usepackage{pdflscape}
\usepackage{steinmetz}
\usepackage{enumerate}
\usepackage{amsmath,amssymb,amsfonts,empheq}
\usepackage{bm}
\usepackage{bbm}
\usepackage{bbold}
\usepackage{dsfont}
\usepackage{mathrsfs}
\usepackage{wasysym}
\usepackage{textcomp}

\usepackage{txfonts}

\newcommand{\fx}{{\bf{x}}}
\newcommand{\fy}{{\bf{y}}}
\newcommand{\fz}{{\bf{z}}}
\newcommand{\fu}{{\bf{u}}}
\newcommand{\fv}{{\bf{v}}}
\newcommand{\fp}{{\bf{p}}}

\newcommand{\beq}{\begin{equation}}
\newcommand{\eeq}{\end{equation}}
\newcommand{\bea}{\begin{eqnarray}}
\newcommand{\eea}{\end{eqnarray}}
\newcommand{\prD}{\mathcal{T}}
\newcommand{\prB}{\mathcal{B}}
\newcommand{\FT}{\mathcal{F}}
\newcommand{\Op}{\mathcal{U}}

\newcommand{\mtx}[1]{{\mathrm{#1}}}

\begin{document}
%
\title{On Landau's eigenvalue theorem and \\ information cut-sets}

\author{\IEEEauthorblockN{Massimo Franceschetti}~\thanks{The author is with the Dept. of Electrical and Computer Engineering
University of California at San Diego
La Jolla, California 92093--0407
Email: massimo@ece.ucsd.edu}}
\maketitle

\begin{abstract}
A variation of Landau's eigenvalue theorem describing the phase transition  of the eigenvalues of a time-frequency limiting, self adjoint operator is presented.
The total number of degrees of freedom of square-integrable,  multi-dimensional, bandlimited functions is defined in terms of Kolmogorov's $n$-width and computed in some limiting regimes  where the original theorem 
cannot  be directly applied. Results are used to characterize up to order  the total amount of information that can be transported in time and space by multiple-scattered electromagnetic waves,  rigorously addressing a question originally posed in the early works of Toraldo di Francia and Gabor. Applications in the context of wireless communication and electromagnetic sensing are discussed.
 \end{abstract}


%
\IEEEpeerreviewmaketitle

\section{Introduction}
How much information can a prescribed electromagnetic waveform transport in time and space? This basic question is of mathematical and physical interest, and has numerous engineering applications, including in communications, sensing, imaging, radar detection and classification systems. Information theory of wireless communication provides a partial answer, but assumes a stochastic model of the fading channel,  taking for granted the physical nature of the quantities which figure in its formalism. In this paper, we take the point of view that communication of information is the transmission of physical states from transmitters to receivers and, as such, it is subject to the laws of nature. We rigorously compute the physical limits for the transport of information by electromagnetic waves in terms of degrees of freedom, and discuss their significance  in an engineering context.

Our results build upon Landau's theorem~\cite{szego}, concerning the asymptotic behavior of the eigenvalues of a certain integral equation arising from the problem of simultaneous concentration of a function and its Fourier transform. This theorem asserts that, in a well defined sense, the eigenvalues undergo a sharp transition from values close to one to values close to zero, and the scale of this transition characterizes the asymptotic dimension of the space of bandlimited functions. 
The problem was originally considered jointly by Landau, Pollak and Slepian in a series of papers~\cite{Slepian1,Slepian2,Slepian3,Slepian4,Slepian5}, of which \cite{SlepianS1} and \cite{SlepianS2} provide excellent tutorial reviews. The precise width of the transition  in the single-dimensional case has been first conjectured  by Slepian~\cite{Slepian}, and finally computed rigorously  by Landau and Widom~\cite{widom}. The theorem we refer here gives  only the first order characterization of the transition, but it describes concentration over arbitrary sets and in arbitrary dimensions. We provide  extensions of the original statement, and apply them  in certain geometric configurations arising in the context of information transport through electromagnetic propagation. 

Consider a wireless network in which a set of users located in one region  of the network wishes to communicate with users located in an another region,  see Figure~\ref{fig:cutsetsp}. The number of channels available for communication is limited by  the dimensionality of the signal's space on the cut-set boundary through which electromagnetic propagation occurs. 
\begin{figure}
\begin{center}
{
\scalebox{.7}{\includegraphics{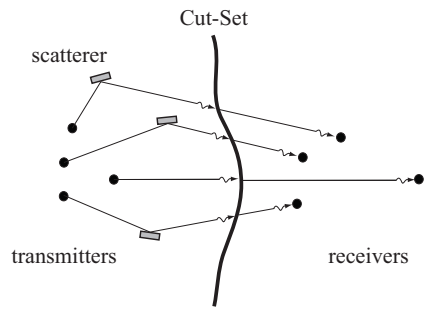}}}
\end{center}
\caption{Parallel channels consume spatial resource along the cut.}
\label{fig:cutsetsp}
\end{figure}
This corresponds to the minimum number of  basis functions required to represent the signal  on the cut-set boundary, and  is independent of the communication scheme employed. It  provides a bound on the amount of spatial  and frequency multiplexing achievable using arbitrary technologies and in arbitrary scattering environments. It essentially corresponds to the \emph{number of degrees of freedom} of the space-time field used for communication and is determined by the phase transition of the eigenvalues referred above.  Similarly, the number of independent features that can be extracted from a scattering system using a  probe signal for remote sensing, imaging, detection and classification systems is  limited by the  number of degrees of freedom of the field and characterized by the same phase transition.

We compute the  number of degrees of freedom of multiple scattered electromagnetic waves  in terms of Kolmogorov's $n$-width, and  provide a physical interpretation in terms of cut-set information flow. Our results extend the single frequency treatments of ~\cite{bucci1,bucci2,Miller1,Miller2,Poon,Kennedy,migliore1,migliore2} to signals of non-zero frequency bandwidth. 
In a broader framework, they rigorously address the question of how much information does an electromagnetic waveform carry in time and space, first posed in  the early works of  Toraldo di Francia~\cite{Toraldo1, Toraldo2} and Gabor~\cite{Gabor1,Gabor2}. 

The number of degrees of freedom of the field is also related to the information capacity of multi-user communication systems. Its application  to bound the Shannon capacity scaling of wireless networks is described in~\cite{migliore1} and \cite{migliore2}, and can be extended to signals with a non-zero frequency band using the results given here.   


The  rest of the paper is organized as follows. The next section provides rigorous statements of the mathematical results.
Applications are discussed in Section~\ref{sec:application}. 
Section~\ref{sec:conc} draws conclusions and mentions some possible future work. The proof of the main theorem
is given in Appendix~\ref{sec:proof}. 
\section{Statement of the results}
We begin by presenting our results in two dimensions, as this  form better suits the application discussed in the next section. A more general statement appears at the end of this section.
Let $P$ and $Q$ be measurable sets in $\mathbb{R}^2$ with boundaries of measure zero. 
For any  point $ \textbf{p} \in \mathbb{R}^2$, and 
positive scalar $\beta>0$, 
we indicate by $\beta (P+ \textbf{p})$ the set of points of the form $\beta(\fx+ \textbf{p})$ with $\fx\in P$. Clearly, we have
\beq
m(\beta P)= \beta^2 m(P),
\eeq 
where $m(\cdot)$ indicates Lebesgue measure.
Similarly, for $\fx=(x_1,x_2)$, we have
\beq
m\{(\beta x_1,x_2):\fx\in P\} =\beta \,m(P).
\eeq
For any two points $\fx=(x_1, x_2)$, $\fu=(u_1, u_2)$,  we let $\fx \cdot \fu = x_1u_1+  x_2 u_2$. For $f \in L^2(\mathbb{R}^2)$,   the Fourier transform of $f$ is
\beq
\FT f =  \int_{\mathbb{R}^2} f(\fx) \exp(- i \fx \cdot \fu) d \fx.
\eeq

Consider now two subspaces of $L^2(\mathbb{R}^2)$ consisting of the functions supported in $P$ and of those  whose Fourier transform is supported in $Q$, namely
\begin{align}
\mathscr{T}_P &= \{f \in L^2 : f(\fx) = 0, \;  \fx \not \in P\}, \\
\mathscr{B}_Q &= \{f \in L^2 : \FT f (\fu)=0, \; \fu \not \in Q\}.
\end{align}

The  orthogonal projections  onto these subspaces are defined using the indicator function 
\begin{equation}
\mathds{1}_Q(\fx) =  \left\{ \begin{array}{ll}
1 & \mbox{if $\fx \in Q$},\\
0 & \mbox{otherwise},\end{array} \right. 
\end{equation}
in the following way
\begin{align}
\prD_P f &= \mathds{1}_P(\fx) f(\fx), \label{eq:rect} \\
\prB_Q f &= \FT^{-1} \mathds{1}_Q \FT f =  \int_{\mathbb{R}^2} h(\fx-\fy) f(\fy) d \fy, \label{eq:sinc}
\end{align}
where 
\beq 
\FT h = \mathds{1}_Q.
\eeq
We are interested in the behavior of the eigenvalues in the integral equation 
\beq
\prD_P \prB_Q \prD_P \varphi(\fx) = \lambda \varphi(\fx),
\label{fred}
\eeq
where the operator  $\prD_P \prB_Q \prD_P$ is positive, self-adjoint, compact, and bounded by one. The eigenvalues are a countable set that characterizes the \emph{asymptotic dimension} of the subspace $\mathscr{B}_Q$. Namely, the number of eigenvalues above level $\epsilon$ corresponds to the dimension of the minimal subspace approximating the elements of $\mathscr{B}_Q$ within $\epsilon$ accuracy over the set $P$. A rigorous definition of asymptotic dimension is given in Section~\ref{sec:application}. To determine this number, Landau~\cite{szego}
considered the case where $Q$ is a fixed set, while $P$ varies over the family $\beta P'$, with $P'$ fixed, and determined  the number of eigenvalues significantly greater than zero as $\beta \rightarrow \infty$. It turns out that the eigenvalues arranged in non-increasing order undergo a transition from values close to one to values close to zero in an interval of indexes around $\beta^2$. The transition  is sharp, of width $o(\beta^2)$, so that the eigenvalues plotted as a function of their indexes appear as a step function  when viewed at the scale of $\beta^2$, see Figure~\ref{fig:eigtran}.  
\begin{figure}
\begin{center}
{
\scalebox{.8}{\includegraphics{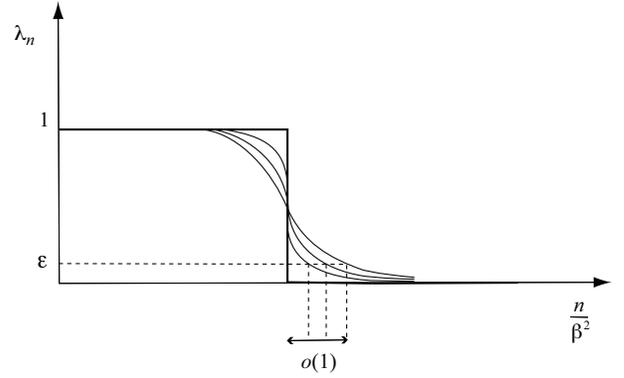}}}
\end{center}
\caption{Phase transition of the eigenvalues.}
\label{fig:eigtran}
\end{figure}
The asymptotic dimension of the space $\mathscr{B}_Q$ is given by the  transition point  of this step function, that  is of the order of $\beta^2$, when the support set $P$ is appropriately scaled by blowing-up all of its coordinates. Due to symmetry,  the asymptotic dimension of the space $\mathscr{T}_P$ can similarly be obtained by considering the operator $\prB_Q \prD_P \prB_Q$ and scaling $Q=\beta Q'$ while keeping $P$ fixed. 
The precise result is stated as follows: 
\newtheorem{thm}{Theorem}
\begin{thm}
\label{th:landau}
(Landau). For any  $0<\epsilon<1$,  let $N_\beta(\epsilon)$ be the number of eigenvalues of $\prD_{\beta P'} \prB_Q \prD_{\beta P'}$ not smaller than $\epsilon$, then we have   
\beq
\lim_{\beta \rightarrow \infty} \frac{N_\beta(\epsilon)}{\beta^2} = (2 \pi)^{-2} m(P') m(Q).
\eeq
\end{thm}

In Landau's case, spectral concentration is achieved by scaling all coordinates  of one of the two  support sets. We prove the analogous result by scaling only one coordinate of both support sets.  While in Landau's case the scaling of the coordinates is uniform, in our case coordinates are scaled at independent rates.  


Let $P$ and $Q$ be two fixed sets. We consider  the following families 
\beq
P_\tau = \{ (\tau x_1, x_2): (x_1,x_2) \in P\} 
\eeq
\beq
Q_\rho = \{ (x_1, \rho x_2): (x_1,x_2) \in Q\}.
\eeq
\begin{thm}
For any  $0<\epsilon<1$,  let $N_{\tau,\rho}(\epsilon)$ be the number of eigenvalues of $\prD_{P_{\tau}} \prB_{Q_{\rho}} \prD_{P_{\tau}}$ not smaller than $\epsilon$, then we have   
\beq
N_{\tau,\rho}(\epsilon)= (2 \pi)^{-2} m(P) m(Q) \tau \rho +o(\tau \rho),
\label{eq:uno1}
\eeq
as $\tau, \rho \rightarrow \infty$.
\label{th1}
\end{thm}
The proof is based on decomposing the operator using orthonormal functions, and then computing an integral that  differs from  the one of Landau, due to the different scaling of the space.
The result also extends to higher dimensions, where any combination of coordinates' scaling in the original or transformed domain can be performed. Our theorem can be stated for any invertible linear mapping of the support sets, subject to a limiting condition.  
Let $P$ and $Q$ be measurable sets in $\mathbb{R}^N$ with boundaries of measure zero. For any real  matrix $\mtx{A}$ of size $N\times N$, we indicate by $\mtx{A}P$ the set of points of the form $\mtx{A} \mtx{x}$, where $\mtx{x}$ is the column vector composed of the elements of $\fx\in P$. We also indicate  with $|\mtx{A}|$  the determinant of $\mtx{A}$, with $\mtx{A^T}$ the transpose of $\mtx{A}$, and with with $U$   the unit ball in $L^2(\mathbb{R}^N).$ We let $\mtx{A}=\mtx{A}(\tau)$ and $\mtx{B}=\mtx{B}(\rho)$ for real parameters $\tau$ and $\rho$. The case in which either matrix is constant, or depends on multiple parameters is completely analogous.
\begin{thm}
For any  $0<\epsilon<1$,  let $N_{\mtx{A},\mtx{B}}(\epsilon)$ be the number of eigenvalues of $\prD_{\mtx{A}P} \prB_{\mtx{B}Q} \prD_{\mtx{A}P}$ not smaller than $\epsilon$.
If 
\beq
\lim_{(\tau, \rho) \rightarrow (\tau_0,\rho_0)} \mtx{B^T}\mtx{A} U = \mathbb{R}^N,
\eeq
then
we have   
\beq
\lim_{(\tau, \rho) \rightarrow (\tau_0,\rho_0)} \frac{N_{\mtx{A},\mtx{B}}(\epsilon)}{|\mtx{A}|\,|\mtx{B}|} = (2 \pi)^{-N} m(P) m(Q).
\eeq
\label{th:tre}
\end{thm}

\newtheorem{rmk}{Remark}
Landau's theorem is a special case where $\mtx{A}$ is a scalar matrix and $\mtx{B}$ is the identity. Its  two-dimensional version is stated in Theorem~\ref{th:landau}. Theorem 
\ref{th1} corresponds to the special case
\beq
\mtx{A}= \left( \begin{array}{cc}
\tau & 0 \\
0 & 1  
\end{array} \right) \; \;
\mtx{B}= \left( \begin{array}{cc}
1 & 0 \\
0 & \rho  
\end{array} \right).
\eeq

\section{Application of the results}
\label{sec:application}
Landau~\cite{szego} does not mention the extensions  to his theorem described above, but no doubt that if asked he could have proved them effortlessly. Nonetheless, these results do not seem to appear anywhere. We believe that our main contribution is to point out their significance  in the context of communication and sensing using electromagnetic waves.

\subsection{The number of degrees of freedom}
When communication occurs through propagation of electromagnetic waves, the effective dimensionality, or \emph{number of degrees of freedom}, of the space-time field is a key information-theoretic quantity related to the capacity of any spatially distributed communication system~\cite{migliore1,migliore2}. This quantity can be computed by evaluating the number of significant eigenvalues in Theorems~\ref{th:landau}, \ref{th1}, or \ref{th:tre}.

Consider the space $\mathscr{S}$ of real space-time waveforms $f(s,t) \in L^2(\mathbb{R}^2)$ and equipped with the norm
\beq
||f||^2 = \left( \iint_{P} f^2(s,t)ds dt \right)^{1/2}. 
\eeq
A subspace $\mathscr{B}_Q \subset \mathscr{S}$ is given by the space of functions $f(s,t) \in L^2(\mathbb{R}^2)$ of  spectral support $Q$. We assume their energy is normalized so that
\beq
\iint_{\mathbb{R}^2} f^2(s,t) ds dt \leq1.
\label{enbound}
\eeq

Given a level of accuracy $\epsilon>0$,  define the number of
degrees of freedom at level $\epsilon$ of the space $\mathscr{B}_Q$ in $\mathscr{S}$
\beq
\mathcal{N}_{\epsilon}(\mathscr{B}_Q) = \min\{n : d_n(\mathscr{B}_Q, \mathscr{S}) \leq \epsilon\},
\eeq
where $d_n(\mathscr{B}_Q, \mathscr{S})$ is the Kolmogorov $n$-width~\cite{Pinkus} of the space $\mathscr{B}_Q$ in $\mathscr{S}$.
Letting  $\mathscr{S}_n$ be an $n$-dimensional subspace of $\mathscr{S}$, this is defined as
\beq
d_n(\mathscr{B}_Q, \mathscr{S}) = \inf_{\mathscr{S}_n \subset \mathscr{S}} D_{\mathscr{S}_n}(\mathscr{B}_Q),
\eeq
where the deviation
\beq
D_{\mathscr{S}_n}(\mathscr{B}_Q) =\sup_{f \in \mathscr{B}_Q} \inf_{g \in \mathscr{S}_n}||f-g||.
\eeq
In words, the deviation represents how well $\mathscr{B}_Q$ may be uniformly approximated by the elements of an $n$-dimensional subspace of $\mathscr{S}$, while the $n$-width is the smallest of such deviations over all $n$-dimensional subspaces of $\mathscr{S}$.
The number of degrees of freedom
$N_\epsilon(\mathscr{B}_Q)$ represents the dimension of the minimal subspace representing the elements of $\mathscr{B}_Q$ within $\epsilon$ accuracy over the set $P$. A basic result in approximation theory (see e.g.\ \cite[Ch. 2,  Prop. 2.8]{Pinkus}) states that
\beq
d_n(\mathscr{B}_Q, \mathscr{S}) = \sqrt{\lambda_n},
\label{basic}
\eeq
where $\lambda_n$ is the $n$-th eigenvalue (arranged in non-increasing order) of  the Fredholm integral equation of the second kind in (\ref{fred}). 
It now follows from~(\ref{basic}) that Theorem~\ref{th1}  allows to compute the number of degrees of freedom by scaling only one of the two coordinates of the space-time field, together with the  transformed version of the other. 

The number of degrees of freedom defined in this way appears to be a principal feature of the mathematical model of the real world of transmitted signals, that is practically insensitive to small changes of a secondary feature of the model, such as the accuracy $\epsilon$ of the measurement apparatus with which the signals are detected. This is evident by rewriting (\ref{eq:uno1})  as
\beq
N_{\tau,\rho}(\epsilon) = (2 \pi)^{-2} m(P) m(Q) \, \tau \rho + o(\tau \rho) \;\;\; \mbox{as } \tau,\rho \rightarrow \infty,
\eeq
where the $\epsilon$-dependence  appears  hidden as a pre-constant of the second order term  $ o(\tau \rho)$  in the phase transition of the eigenvalues.  

\subsection{Cut-set of two-dimensional circular domains}
Theorems~\ref{th:landau} and \ref{th1} can be applied to  evaluate the number of degrees of freedom of waveforms of a given frequency band over spatial cut-sets separating transmitters and receivers in a wireless communication setting, or separating scattering objects from sensing devices in imaging and sensing systems. 

Consider the case  of a two-dimensional domain of  cylindrical symmetry, 
\begin{figure}
\begin{center}
{
\scalebox{.95}{\includegraphics{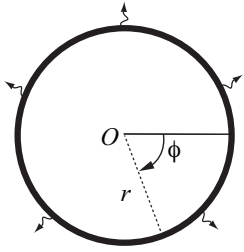}}}
\end{center}
\caption{Cut-set boundary of two-dimensional circular  domain.}
\label{fig:radiator}
\end{figure}
in which an electromagnetic field is radiated by a configuration of currents located inside a circular domain of radius $r$ of an arbitrary scattering environment, oriented perpendicular to the domain, and constant along the direction of flow. This can model, for example, an arbitrary scattering environment where a spatial distribution of wireless transmitters is placed inside the circular domain,  and
communicate to receivers placed outside the domain. The same model applies to a remote sensing system where  objects  inside the domain are illuminated by an external waveform, and the scattered field is recovered by sensors placed outside the domain. The radiated field away from the scattering system and measured at the receivers is completely determined by the field on the cut-set boundary   through which it propagates, see Figure~\ref{fig:radiator}.  On this boundary, we can refer to a  scalar field  $f(\phi,t)$ that is a function  of only two scalar variables: one angular and one temporal. The corresponding four field's representations, linked by Fourier transforms, are depicted in Figure~\ref{fig:repr1}, where the angular frequency $\omega$ indicates the transformed coordinate of the time variable $t$, the wavenumber $w$ indicates the transformed coordinate of the angular variable $\phi$, and $\mathcal{F}f(\phi,t) = \widehat{F}(w,\omega)$.
\begin{figure}
\begin{center}
\begin{picture}(140,100)
\put(0,0){$\widehat{f}(w,t)$}
\put(0,80){$f( \phi,t)$}
\put(100,80){$F( \phi,\omega)$}
\put(100,0){$\widehat{F}( w,\omega)$}
\put(15,40){\vector(0,1){30}}
\put(15,40){\vector(0,-1){30}}
\put(68,0){\vector(1,0){20}}
\put(68,0){\vector(-1,0){30}}
\put(68,80){\vector(1,0){20}}
\put(68,80){\vector(-1,0){30}}
\put(115,40){\vector(0,1){30}}
\put(115,40){\vector(0,-1){30}}
\end{picture}
\end{center}
\caption{Field's representations linked by Fourier transforms.}
\label{fig:repr1}
\end{figure}
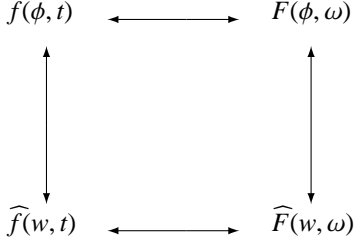

For convenience, we let $r$ be normalized by the speed of light, so that both radius-frequency and time-frequency products are dimensionless quantities.
We wish to determine the  number of degrees of freedom of the space-time field $f(\phi,t)$ on the cut-set boundary,  that is radiated for $t\in[-T/2,T/2]$ and occupies a bandwidth $\omega \in [-\Omega,\Omega]$; it is observed over the interval $\phi \in [-\pi,\pi]$ and occupies a  wavenumber bandwidth $w\in[-W,W]$. Of course, these statements make little mathematical sense since a field satisfying above constraints would have bounded support in both the natural and transformed domain, so it would be zero everywhere. To make our considerations precise, we need to take appropriate scaling limits.

To determine the correct scaling laws for the support sets of the field, we introduce some geometric constraints. On the one hand, we limit the observation domain  to $\phi \in [-\pi,\pi]$, so that the spatially periodic field on the cut-set boundary is observed over a single period. On the other hand, 
it is well known~\cite{bucci1}  that in this case the wavenumber bandwidth is related to the frequency of transmission in such a way that for any possible configuration of sources and scatterers inside the circular radiating domain, we have at most
\beq
 w =  \omega r +o(\omega r) \;\;\; \mbox{as } \omega r \rightarrow \infty.
\label{circular}
\eeq
It follows that we need to scale the support sets of the field  subject to the constraints  depicted in Figure~\ref{fig:PQ}.
\begin{figure}
\begin{center}
{
\scalebox{.7}{\includegraphics{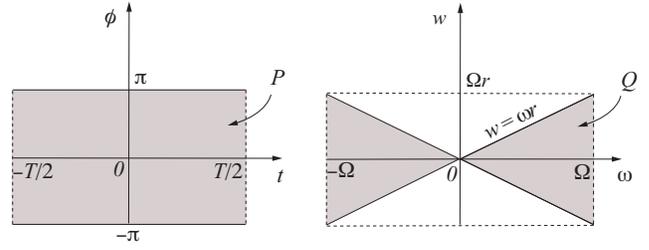}}}
\end{center}
\caption{Support sets in the natural and transformed domains.}
\label{fig:PQ}
\end{figure}
We can blow-up the spectral support $Q$ while keeping $P$ fixed by letting $\Omega \rightarrow \infty$. In this case, by Theorem~\ref{th:landau}   the number of degrees of freedom  over a fixed transmission time  and cut-set  interval of width $2 \pi$, in the wide band frequency regime is
\begin{align}
\mathcal{N}_\epsilon(\mathscr{B}_Q) &= \frac{2 \Omega^2 r T 2 \pi}{(2\pi)^2} + o(\Omega^2) \nonumber \\
&= \frac{\Omega T}{\pi}  \frac{2 \pi r \Omega}{2\pi} +o(\Omega^2)  \;\;\; \mbox{as } \Omega \rightarrow \infty.
\label{eq:uno}
\end{align}
On the other hand, our geometric configuration does not allow to blow-up the support $P$ while keeping $Q$ fixed because the cut-set domain is limited to an angle $2\pi$. Landau's theorem in its original form is of no help here, but we can apply Theorem~\ref{th1}  to obtain the number of degrees of freedom over a fixed frequency band, by scaling the time coordinate of $P$ and  the coordinate of $Q$ corresponding to the size of the radiating system, and we have
\begin{align}
\mathcal{N}_\epsilon(\mathscr{B}_Q) &= \frac{\Omega T}{\pi}  \frac{2 \pi r \Omega}{2\pi} + o(rT)  \;\;\; \mbox{as } T,r \rightarrow \infty .
\label{eq:due}
\end{align}

Equations (\ref{eq:uno}) and (\ref{eq:due}) show that the number of degrees of freedom is given by the product of two factors, each viewed in an appropriate asymptotic regime: one accounting for the number of degrees of freedom in the time-frequency domain, $\Omega T/\pi$, and  another accounting for the number of degrees of freedom in the space-wavenumber domain, $2 \pi r \Omega  /(2 \pi)$. The latter factor physically corresponds to the perimeter of the disc of radius $r$ normalized by an interval of wavelengths $2 \pi/\Omega$,  and can be interpreted as the spatial cut-set through which the information must flow. 
The idea is that for any finite size system, the wavenumber bandwidth is a limited resource. Each parallel channel occupies a certain amount of spatial resource on the cut, proportional to the wavelength of transmission, and these channels must be sufficiently spaced along the cut for the corresponding waveforms to provide independent streams of information. The total number of channels  is then given by the total spatial resource, given by the cut-length $2 \pi r$, divided by the total occupation cost, given by the wavelength-interval $2 \pi/\Omega$.


\subsection{Comparison with the single-dimensional case}
\label{sec:comparison}
The number of time-frequency degrees of freedom of any  signal $f(t) \in L^2(\mathbb{R})$  bandlimited to $[-\Omega,\Omega]$, or  timelimited to $[-T/2,T/2]$,   is given by the time-bandwidth product
\beq
\mathcal{N}_\epsilon(\mathscr{B}_{\Omega}) = \frac{\Omega T}{\pi} +o(\Omega T)\;\;\; \mbox{as } \Omega T \rightarrow \infty.
\label{comb1}
\eeq
In this case, spectral concentration occurs by scaling either the transmission time $T$, or the frequency band $\Omega$. The work in~\cite{widom} gives the precise asymptotic order of the term $o(\Omega T)$, which is $\log (\Omega T) \log (1/\epsilon)$. 

Similarly, letting $W= \omega r+o(\omega r)$, the number of space-wavenumber degrees of freedom of any  signal $f(\phi) \in L^2[-\pi,\pi]$ radiated with frequency $\omega$ from the interior of a circular domain of radius $r$ and observed on the circular perimeter boundary of angle $\phi \in[-\pi,\pi]$,  is given by the space-bandwidth product 
\beq
\mathcal{N}_\epsilon(\mathscr{T}_{2\pi}) = \frac{\omega 2 \pi r}{\pi} +o(\omega r) \;\;\; \mbox{as } r \omega \rightarrow \infty.
\label{comb2}
\eeq
In this case, spectral concentration occurs
by scaling either the size of the radiating system $r$, or the frequency of transmission $\omega$. 

The results (\ref{eq:uno}) and (\ref{eq:due}) can be viewed as a combination of (\ref{comb1}) and (\ref{comb2}).
An heuristic way  to compute the total number of degrees of freedom of the space-time field $f(t,\phi)$ would be to integrate   over frequencies while accounting for the number of spatial degrees of freedom that every frequency component carries. This corresponds to 
simply integrate (\ref{comb2}) over the  frequency bandwidth and, according to (\ref{comb1}), multiply the result by $T/\pi$
\begin{align}
\mathcal{N}_\epsilon(\mathscr{B}_Q) &= \frac{T}{\pi} \int_{0}^{\Omega} \frac{\omega 2 \pi r}{\pi} d \omega, \nonumber \\
&= \frac{\Omega T}{\pi}  \frac{2 \pi r \Omega}{2\pi},
\label{eq:b}
\end{align}
which gives the correct result.  Of course, the problem of this heuristic is clear: it does not account for the possible accumulation of the error $\epsilon$ when computing the integral along the frequency spectrum. Thus, the need for the rigorous method presented in this paper arises.

\subsection{Cut-set of three-dimensional spherical domains}
We extend results to three dimensions by considering a spherical radiating system of radius $r$. In this case, the  surface of the sphere is interpreted as a cut-set  through which the information must flow and provides a limit on the amount of information that can radiate from the interior of the domain to the outside space. Each scalar component $f(\phi_1,\phi_2,t)$ of the vector field on the cut-set boundary is a function of two angular coordinates identifying a point on the surface and one temporal one, see Figure~\ref{fig:sphere}. 
\begin{figure}
\begin{center}
{
\scalebox{.55}{\includegraphics{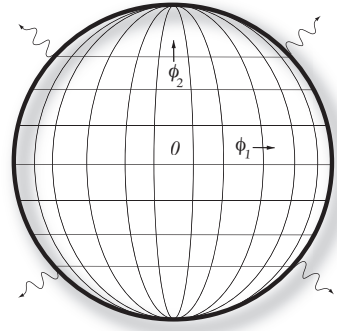}}}
\end{center}
\caption{Cut-set boundary of three-dimensional spherical  domain.} 
\label{fig:sphere}
\end{figure}
The support sets of each scalar field component are the ones depicted in Figure~\ref{fig:PQ1}, where we indicate with $w_1$ the transformed coordinate of the variable $\phi_1$, with $w_2$ the transformed coordinate of the variable $\phi_2$, and with $\Phi$ the solid angle subtended by the surface boundary at the center of the sphere.
%
\begin{figure}
\begin{center}
{
\scalebox{.6}{\includegraphics{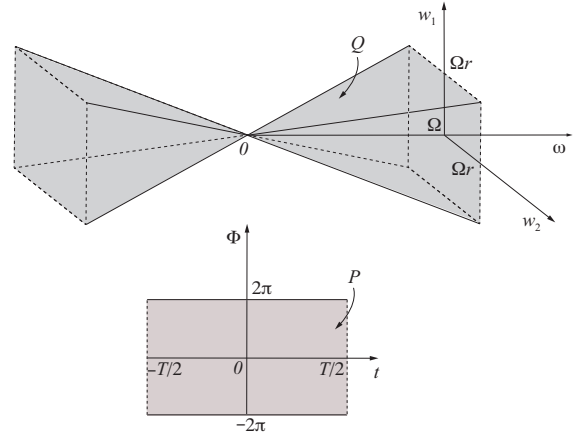}}}
\end{center}
\caption{Support sets in the natural and transformed domains.}
\label{fig:PQ1}
\end{figure}

We then have
\beq
m(P)= 4 \pi T,
\label{sub1}
\eeq
\beq m(Q)= \frac{8}{3}\Omega^3 r^2.
\label{sub2}
\eeq
By Theorem~\ref{th:tre}, with Landau's type scaling
\beq
\mtx{A}= \left( \begin{array}{ccc}
1 & 0 & 0 \\
0 & 1  & 0 \\
0 & 0  &  1
\end{array} \right), \; \;
\mtx{B}= \left( \begin{array}{ccc}
\rho & 0 & 0 \\
0 & \rho & 0 \\
0& 0 & \rho  
\end{array} \right),
\label{landauscale}
\eeq
and using (\ref{sub1}) and (\ref{sub2}),  it follows that the number of degrees of freedom in the wide band frequency regime is
\beq
\mathcal{N}_\epsilon(\mathscr{B}_Q) = 4 \pi r^2 \frac{\Omega^3 T}{3 \pi^3} +o(\Omega^3) \;\;\; \mbox{as } \Omega \rightarrow \infty,
\label{eq:rigor1}
\eeq
where $\Omega = \rho \Omega'$ with $\Omega'$ fixed and $\rho \rightarrow \infty$.
With the alternative scaling
\beq
\mtx{A}= \left( \begin{array}{ccc}
\tau & 0 & 0 \\
0 & 1  & 0 \\
0 & 0  &  1
\end{array} \right), \; \;
\mtx{B}= \left( \begin{array}{ccc}
1 & 0 & 0 \\
0 & \rho & 0 \\
0& 0 & \rho  
\end{array} \right),
\label{alternativescale}
\eeq
we also have that the number of degrees of freedom over a fixed frequency band for large radiating systems and transmission time is
\beq
\mathcal{N}_\epsilon(\mathscr{B}_Q) = 4 \pi r^2 \frac{\Omega^3 T}{3 \pi^3} +o(r^2 T) \;\;\; \mbox{as } T,r \rightarrow \infty, 
\label{eq:rigor2}
\eeq
where $T=\tau T'$, $r= \rho r'$, with $T',r'$ fixed and $\tau, \rho \rightarrow \infty$.

Above rigorous results can also be obtained using the heuristic method described in Section~\ref{sec:comparison}.
The number of space-wavenumber degrees of freedom of the  electromagnetic field radiated with frequency $\omega$ from the interior of a spherical domain of radius $r$ and observed  on the surface boundary of solid angle $\Phi \in[-2\pi,2\pi]$, is given by the space-bandwidth product~\cite{bucci1, bucci2}
\beq
\mathcal{N}_\epsilon(\mathscr{T}_{4\pi})=\frac{\omega^2 4 \pi r^2}{\pi^2} + o[(\omega r)^2] \;\;\; \mbox{as } r \omega \rightarrow \infty.
\label{eq:narrow2}
\eeq
Integrating over the frequency bandwidth and multiplying the result by $T/\pi$, we obtain 
\begin{align}
\mathcal{N}_\epsilon(\mathscr{B}_Q) &=  \frac{4 \pi r^2}{\pi^2} \frac{T}{\pi} \int_{0}^{\Omega} \omega^2 d \omega \nonumber \\
& = 4 \pi r^2 \frac{\Omega^3 T}{3 \pi^3},  
\label{h}
\end{align}
which is consistent with  the rigorous results in (\ref{eq:rigor1}) and (\ref{eq:rigor2}).

\subsection{Cut-set of general rotationally symmetric domains}
Results can be further generalized  considering a radiating system enclosed in a convex domain bounded by a surface with rotational symmetry. Consider a  cylindrical coordinate system $(r,\phi, z)$, a closed analytic curve $\zeta=\zeta(r,z)$ lying in the plane $\phi=0$ and symmetric with respect to the $z$ axis, and the surface of revolution obtained by rotating the curve about the $z$ axis, see Figure~\ref{fig:oblate}.  
\begin{figure}
\begin{center}
{
\scalebox{.8}{\includegraphics{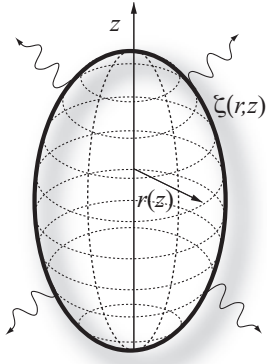}}}
\end{center}
\caption{Cut-set boundary of three-dimensional rotationally symmetric domain.}
\label{fig:oblate}
\end{figure}
In this case, we can choose a pair of  coordinates $(\phi_1,\phi_2)$ on the surface such that any meridian curve covers a $2 \pi$ range, and the spatial bandwidth along the meridian is constant and  bounded by~\cite{savarese} 
\beq
w_1= \omega \frac{\ell}{2 \pi}  +o(\omega \ell) \;\; \mbox{as } \omega \ell \rightarrow \infty,
\eeq
where $\ell$ is the Euclidean length of the curve, normalized to the speed of light so that $\omega \ell$ is dimensionless. This should be compared with the geometric constraint for circular curves given in (\ref{circular}). On the other hand, any latitude line is a circle of radius $r(z)$ that also covers a $2 \pi$ range, and the spatial bandwidth along this line is bounded by
\beq
w_2 = \omega r(z) + o(\omega r), \;\; \mbox{as } \omega r \rightarrow \infty.
\eeq
It follows that while the domain $P$ covers a solid angle $4 \pi$, the domain $Q$ varies along $z$ according to the meridian curve parametrization $\zeta=\zeta(r,z)$, and we have
\beq
m(P)=4 \pi T,
\label{eq:sub11}
\eeq
\begin{align}
m(Q) &= \frac{8 \Omega^3}{3} \ \frac{\ell }{2 \pi} \frac{\pi}{2} \int_\zeta r(z) dz  \nonumber \\
&= \frac{8 \Omega^3}{3}  \frac{\ell }{4 \pi}  \int_\zeta  \pi r(z)  dz  \nonumber \\
&= \frac{8 \Omega^3}{3} \frac{\mathcal{A}}{4 \pi},
\label{eq:sub22}
\end{align}
where $\mathcal{A}$ is the surface area of the radiating volume. 

By Theorem~\ref{th:tre},  the scaling in (\ref{landauscale}), and using (\ref{eq:sub11}) and (\ref{eq:sub22}), we have that the number of degrees of freedom in the wide band frequency regime is
\beq 
\mathcal{N}_\epsilon(\mathscr{B}_Q) = \frac{\mathcal{A}  \Omega^3 T}{3 \pi^3} +o(\Omega^3) \;\;\; \mbox{as } \Omega \rightarrow \infty,
\eeq
where $\Omega = \rho \Omega'$ with $\Omega'$ fixed and $\rho \rightarrow \infty$.
The analogous result is obtained by  letting the transmission time $T= \tau T'$, with $T'$ fixed and $\tau \rightarrow \infty$, and blowing-up all coordinates of the radiating volume, so that both $\ell = \rho \ell'$ and $r= \rho r'(z)$ tend to infinity with $\ell'$ and $r'$ fixed and $\rho \rightarrow \infty$. In this case, 
by Theorem~\ref{th:tre} with the scaling in (\ref{alternativescale}), and using (\ref{eq:sub11}) and (\ref{eq:sub22}), we have
\beq
\mathcal{N}_\epsilon(\mathscr{B}_Q) = \frac{\mathcal{A}  \Omega^3 T}{3 \pi^3} +o(\mathcal{A}T) \;\;\; \mbox{as } T, \mathcal{A} \rightarrow \infty, 
\eeq
where $\mathcal{A}= \rho^2 \mathcal{A'}$, with $\mathcal{A}'$ fixed and $\rho \rightarrow \infty$.
Finally, we can check that the result is consistent with the heuristic calculation:  starting with the 
 number of space-wavenumber degrees of freedom per angular frequency $\omega$~\cite{migliore2, savarese} 
 \beq
 \mathcal{N}_\epsilon(\mathscr{T}_{4 \pi}) = \frac{\mathcal{A} \omega^2}{\pi^2}, 
 \label{h1}
 \eeq
integrating over the frequency bandwidth and multiplying by $T/\pi$, the result follows.
%

%
\subsection{Degrees of freedom of modulated signals}
The results  can be extended to handle the case of modulated signals. Consider  the case of a real signal modulating a sinusoid of carrier  frequency $\omega_c$, occupying a bandwidth $\Delta \omega =[\omega_1,\omega_2]$ centered around $\omega_c$, see Figure~\ref{fig:modulation}. 
\begin{figure}
\begin{center}
{
\scalebox{.7}{\includegraphics{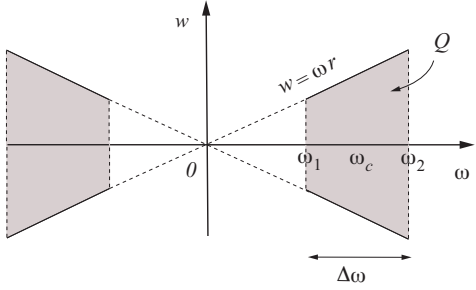}}}
\end{center}
\caption{Support set of a modulated signal in the transformed domain.}
\label{fig:modulation}
\end{figure}
In the two-dimensional case, 
following the same procedure of the previous sections, letting $T=\tau T'$, $r = \rho r'$, and $\rho, \tau \rightarrow \infty$, we obtain
\begin{align}
\mathcal{N}_\epsilon(\mathscr{B}_Q) &=  \frac{T}{\pi}  \frac{2 \pi r}{\pi} \frac{(\omega_2^2-\omega_1^2)}{2} +o(r  T)  \nonumber\\
&=   \frac{T \Delta \omega}{\pi}  \frac{2 \pi r  \omega_c}{\pi} + o(r  T),
\label{f1}
\end{align} 
as $r,T \rightarrow \infty.$

On the other hand, letting
$\Delta \omega = \rho \Delta \omega'$, $\omega_c = \rho \omega_c'$, and $\rho \rightarrow \infty$, we have
\beq
\mathcal{N}_\epsilon(\mathscr{B}_Q) = \frac{T \Delta \omega}{\pi}  \frac{2 \pi r  \omega_c}{\pi} + o(\omega_c \Delta \omega),
\eeq
as $\omega_c, \Delta \omega \rightarrow \infty$.
The number of degrees of freedom is proportional to the time-bandwidth product and to the  cut-set boundary normalized by  the radiating carrier.  

For spherical domains, letting $T = \tau T'$, $r= \rho r'$, and $\tau, \rho \rightarrow \infty$, we have
\begin{align}
\mathcal{N}_\epsilon(\mathscr{B}_Q) &= \frac{T}{\pi}  \frac{4 \pi r^2}{\pi^2}  \frac{(\omega_2^3-\omega_1^3)}{3} + o(T r^2) \nonumber \\
& = \frac{T \Delta \omega}{\pi}   \frac{4 \pi r^2}{\pi^2}     \frac{(\omega_{2}^2+\omega_{1}^2 + \omega_1 \omega_2)}{3} + o(T r^2) \nonumber \\
& =  \frac{T \Delta \omega}{\pi}   \frac{4 \pi r^2 \omega_c^2}{\pi^2}  (1+\delta) + o(T r^2),
\label{f2}
\end{align}
as $T,r \rightarrow \infty$, where  $\delta$ is a constant that depends only on the ratio $\Delta \omega/\omega_c$. Analogously,  letting $\Delta \omega = \rho \Delta \omega'$, $\omega_c = \rho \omega_c'$,  $\rho \rightarrow \infty$, we have
\begin{align}
\mathcal{N}_\epsilon(\mathscr{B}_Q) & =  \frac{T \Delta \omega}{\pi}   \frac{4 \pi r^2 \omega_c^2}{\pi^2}  (1+\delta) + o(\omega_c^2 \Delta \omega ),
\label{f3}
\end{align}
as $\omega_c, \Delta \omega \rightarrow \infty$.  

Results (\ref{f2}) and (\ref{f3}) can also be combined using the scaling matrices 
\beq
\mtx{A}= \left( \begin{array}{ccc}
\tau & 0 & 0 \\
0 & 1  & 0 \\
0 & 0  &  1
\end{array} \right), \; \;
\mtx{B}= \left( \begin{array}{ccc}
\beta & 0 & 0 \\
0 & \rho \beta & 0 \\
0& 0 & \rho  \beta
\end{array} \right),
\label{combine}
\eeq
and letting $\tau \beta, \rho \beta \rightarrow \infty$, obtaining 
\begin{align}
\mathcal{N}_\epsilon(\mathscr{B}_Q) & =  \frac{T \Delta \omega}{\pi}   \frac{4 \pi r^2 \omega_c^2}{\pi^2}  (1+\delta) + o(r^2   \omega_c^2 \Delta \omega  T),
\label{f4}
\end{align}
as $r \omega_c,  \Delta \omega  T\rightarrow \infty$.

Finally, for general rotationally symmetric domains, using (\ref{combine})  we have 
\beq
\mathcal{N}_\epsilon(\mathscr{B}_Q)  =  \frac{T \Delta \omega}{\pi}   \frac{\mathcal{A} \omega_c^2}{\pi^2}  (1+\delta) + o(\mathcal{A}  \omega_c^2  \Delta \omega  T),
\label{f5}
\eeq
as $\mathcal{A} \omega_c^2, \Delta \omega T \rightarrow \infty$, where $\mathcal{A}= \rho^2 \mathcal{A}'$ with $\mathcal{A}'$ fixed.

For narrowband signals $\Delta \omega/\omega_c \ll 1$, and the constant $\delta$ can be made arbitrarily small,  so that the number of degrees of freedom in three dimensions is essentially given by the first term of (\ref{f5}),
which is the natural extension of  of the single-frequency result in~\cite{migliore2, savarese} and reported in  (\ref{h1}), accounting for a non-zero frequency band around frequency $\pm \omega_c$.
It follows that   the number of degrees of freedom per unit time and per unit frequency band is essentially given by
\beq
N_0=\frac{\mathcal{A} \omega_c^2}{\pi^2}.
\label{eq:highenergy}
\eeq

\section{Concluding remarks}
\label{sec:conc}
A key insight of our analysis is that the amount of information, in terms of degrees of freedom, scales with the surface boundary, rather than with the volume of the space. Clearly,  field theory allows for a much larger number of possible field configurations \emph{inside} the radiating volume. Think for example of the number of standing waves insides a black body. This grows with the volume of the space rather than with the surface area. Our results apply to the information that one can gather about these volumetric configurations, from the point of view of an external observer, through electromagnetic propagation. They  pose an information-theoretic limit in terms of degrees of freedom at the scale of the surface boundary through which the field propagates.
Physically, this is due to the Green's propagation operator relating source currents to the radiated field, that essentially behaves as a spatial filter,  projecting the number of observable field configurations  onto a lower dimensional space~\cite{bucci1}. The resulting wavenumber bandlimitation of the field  dictates the geometric constraints on the support sets depicted in Figures~\ref{fig:PQ}, \ref{fig:PQ1}, and \ref{fig:modulation}, leading to our results. The conclusion is that Nature  ``hides'' three-dimensional field configurations, and the world appears to an external observer as having only an apparent three-dimensional informational structure, subject to a two-dimensional representation. Like for Plato's prisoners in the cave, ``the truth would be literally nothing  but the shadows of the images.''

Another important observation is that for any fixed size system the amount of information scales with the radiated frequency. We can increase the number of degrees of freedom as high as we want by transmitting signals modulating larger and larger  frequencies.  In wireless networks, this has been used to identify regimes of linear capacity scaling with the number of users~\cite{kaist, ozgur}.  In the context of electromagnetic imaging,  this allows to increase the spatial resolution of the constructed image by increasing the illumination frequency.  

Our results are based on a variation of Landau's theorem that allows to compute the number of degrees of freedom of square integrable, bandlimited fields in terms of Kolmogorov's $n$-width, using an alternative scaling of the space.  When degrees of freedom of electromagnetic signals are evaluated along a spatial cut-set boundary that separates transmitters and receivers in a wireless network, or between radiating elements and sensing devices in an electromagnetic remote sensing system, our results yield the effective number of parallel channels available through the cut-set boundary in the time-frequency and the space-wavenumber domain. Thus, they provide a bound on the amount of spatial  and frequency multiplexing achievable using arbitrary technologies and in arbitrary scattering environments. 

Relations between the number of degrees of freedom studied here and the Kolmogorov's $\epsilon$-entropy and $\epsilon$-capacities is well known, and follow from  the application of Mityagin's theorem~\cite{lorentz}. Usage of the number of degrees of freedom  to bound the Shannon capacity scaling of wireless networks are described in~\cite{migliore1} and \cite{migliore2}. Our results are limited to linear scalings of the support sets.
Extensions  to  non-linear  scalings suitable to describe signals with sparse supports would also be of interest. In this case, a  non-linear mapping of the support sets would need to achieve spectral concentration while retaining a structure composed of many vanishingly small subdomains.  The work in~\cite{stark} provides one step in this direction, but it is limited to the study of the zeroth order eigenvalue, rather than the whole phase transition of the eigenvalues. More generally, one could  study spectral concentration under different  structural constraints on the support sets, beside bandlimitation.  In this context,  connections with undersampled signal representations~\cite{donoho1} and compressed sensing~\cite{donoho2} for electromagnetic applications need to be explored.


\appendix
\subsection{Proof of Theorem~\ref{th:tre}}
\label{sec:proof}
We Let
$K_{\mtx{A},\mtx{B}}(\fx,\fy)$ be the kernel of the operator $\Op_{\mtx{A},\mtx{B}}  = \prD_{\mtx{A}P} \prB_{\mtx{B}Q} \prD_{\mtx{A}P}$ having eigenvalues $\{\lambda_k\}$.
We have
\beq
K_{\mtx{A},\mtx{B}}(\fx,\fy) = \mathds{1}_{\mtx{A}P}(\fx)  \mathds{1}_{\mtx{A}P}(\fy) h_{\mtx{B}}(\fx-\fy), 
\eeq
so that 
\beq
K_{\mtx{A},\mtx{B}}(\fx,\fx) =\mathds{1}_{\mtx{A}P}(\fx) h_{\mtx{B}}(0).
\label{eq:oft}
\eeq
All is required is to establish the following two lemmas, as they  imply the statement of the theorem using  a standard argument identical to the one in~\cite{szego}. A sketch of the argument is as follows. The two lemmas state that all eigenvalues must be either close to one or close to zero, since both their sum and the sum of their squares have the same scaling order. The sum then essentially corresponds to the number of non-zero eigenvalues and is of the order  of  $|\mtx{A}| |\mtx{B}|$.  
\newtheorem{lem}{Lemma}
\begin{lem}
 $\sum_k \lambda_k = |\mtx{A}|\,|\mtx{B}| \, (2 \pi)^{-N} m(P) m(Q)$.
\label{l1}
\end{lem}
\begin{lem}
$\sum_k \lambda_k^2= |\mtx{A}|\, |\mtx{B}| \, (2 \pi)^{-N} m(P) m(Q) +o(|\mtx{A}| \, |\mtx{B}|)$.
\label{l2}
\end{lem}
\begin{proof}[Proof of Lemma~\ref{l1}]
By Mercer's theorem there exists an orthonormal basis set $\{ \phi_k\}$ for $L^2(\mathbb{R}^N)$, such that 
\beq
K_{\mtx{A},\mtx{B}}(\fx,\fy) = \sum_k \lambda_k \phi_k(\fx)  \phi_k( \fy).
\eeq
By orthonormality and (\ref{eq:oft}), we have
\beq
 \int_{\mathbb{R}^N} \mathds{1}_{\mtx{A}P}(\fx) h_{\mtx{B}}(0) d\fx = \sum_k \lambda_k,
\eeq
and performing the computation 
\begin{align}
 \int_{\mathbb{R}^N} \mathds{1}_{\mtx{A}P}(\fx) h_{\mtx{B}}(0) d\fx &=h_{\mtx{B}}(0) m(\mtx{A}P) \nonumber\\
&= (2 \pi)^{-N} m(\mtx{B}Q) m(\mtx{A}P) \nonumber \\
&= (2 \pi)^{-N} |A| \; |B| \; m(Q) m(P) \nonumber \\
&= \sum_k \lambda_k,
\end{align}
establishes Lemma~\ref{l1}.
\end{proof}
\begin{proof}[Proof of Lemma~\ref{l2}]
We let $K^{(2)}_{\mtx{A},\mtx{B}}(\fx,\fy)$ be the kernel of the operator $\Op^2_{\mtx{A},\mtx{B}}=(\prD_{\mtx{A}P} \prB_{\mtx{B}Q} \prD_{\mtx{A}P})^2$ having eigenvalues $\{\lambda_k^2\}$. We have
\begin{align}
K_{\mtx{A},\mtx{B}}^{(2)}(\fx,\fy) &= \int_{\mathbb{R}^N} K_{\mtx{A},\mtx{B}}(\fx,\fz) K_{\mtx{A},\mtx{B}}(\fz,\fy) d\fz \nonumber \\
&= \int_{\mathbb{R}^N}  \mathds{1}_{\mtx{A}P}(\fx) \mathds{1}_{\mtx{A}P}(\fz) h_{\mtx{B}}(\fx-\fz) \nonumber \\ 
&  \mathds{1}_{\mtx{A}P}(\fz) \mathds{1}_{\mtx{A}P}(\fy) h_{\mtx{B}}(\fz-\fy) d\fz  \nonumber \\
&=  \mathds{1}_{\mtx{A}P}(\fx)   \mathds{1}_{\mtx{A}P}(\fy) \int_{\mtx{A}P} h_{\mtx{B}}(\fx-\fz) h_{\mtx{B}}(\fz-\fy) d\fz. 
\label{eq:intermediate}
\end{align}
By Mercer's theorem there exists an orthonormal basis set $\{ \psi_n\}$ for $L^2(\mathbb{R}^N)$, such that 
\beq
K_{\mtx{A},\mtx{B}}^{(2)}(\fx,\fy) =  \sum_k \lambda_k^2 \psi_k(\fx)  \psi_k( \fy).
\eeq
By orthonormality, we have
\beq
\int_{\mathbb{R}^N} K^{(2)}_{\mtx{A},\mtx{B}}(\fx,\fx) d\fx = \sum_k  \lambda_k^2.
\label{eq:est}
\eeq
By (\ref{eq:intermediate}) it follows that
\beq
\int_{\mathbb{R}^N} K^{(2)}_{\mtx{A},\mtx{B}}(\fx,\fx) d\fx = \int_{\mtx{A}P}\int_{\mtx{A}P} |h_{\mtx{B}}(\fx-\fy)|^2 d \fx d\fy. 
\eeq
We apply the change of variable $\fx =\mtx{A} \fp$, obtaining
\begin{align}
\int_{\mathbb{R}^N} K^{(2)}_{\mtx{A},\mtx{B}}(\fx,\fx) d\fx  &= |\mtx{A}| \,  \int_{\mtx{A} P} d\fy  \int_{P}  |h_{\mtx{B}}(\mtx{A} \fp - \fy)|^2 d\fp\nonumber \\
&=  |\mtx{A}| \,  \int_{P} d\fp  \int_{\mtx{A} P}  |h_{\mtx{B}}(\mtx{A} \fp - \fy)|^2 d\fy. 
\end{align}
We apply another change of variable  $\fu= \mtx{A} \fp - \fy$, obtaining
\begin{align}
\int_{\mathbb{R}^N} K^{(2)}_{\mtx{A},\mtx{B}}(\fx,\fx) d\fx &= |\mtx{A}| \int_{P}d\fp   \int_{\mtx{A} (\fp-P)}   |h_{\mtx{B}}(\fu)|^2 d\fu.
\label{eq:K}
\end{align}
Substituting (\ref{eq:K}) into (\ref{eq:est}) and dividing by $|\mtx{A}| \, |\mtx{B}|$, we have
\beq 
\frac{1}{|\mtx{A}| \, |\mtx{B}|}  \sum_k \lambda^2_k = \int_{P}   F_{\mtx{A},\mtx{B}}(\fp) d\fp,
\eeq
where
\begin{align}
F_{\mtx{A},\mtx{B}}(\fp) &= |\mtx{B}|^{-1}  \int_{\mtx{A} (\fp-P)}   |h_\mtx{B} (\fu)|^2 d\fu.
\label{eq:q}
\end{align}
The function $F_{\mtx{A},\mtx{B}}(\fp)$ is dominated as 
\begin{align}
F_{\mtx{A},\mtx{B}}(\fp) &\leq |\mtx{B}|^{-1} \int_{\mathbb{R}^N}  |h_{\mtx{B}}(\fu)|^2 d\fu \nonumber \\
&= (2 \pi)^{-N}  |\mtx{B}|^{-1}m(\mtx{B}Q) \nonumber \\
&= (2 \pi)^{-N} m(Q), 
\end{align}
that is integrable over $P$. Next, we show that
\beq
\lim_{(\tau,\rho) \rightarrow (\tau_0,\rho_0)}F_{\mtx{A},\mtx{B}}(\fp) = (2 \pi)^{-N} m(Q),
\label{eq:final}
\eeq
 so that by Lebesgue's dominated convergence theorem, we have
\begin{align}
\lim_{(\tau,\rho) \rightarrow (\tau_0,\rho_0)} \frac{1}{|\mtx{A}| \, |\mtx{B}|}   \sum_k \lambda^2_k  &=  \int_{P} \lim_{(\tau,\rho) \rightarrow (\tau_0,\rho_0)} F_{\mtx{A},\mtx{B}}(\fp) d\fp \nonumber \\
&= (2 \pi)^{-N} m(P) m(Q), 
\end{align}
establishing Lemma~\ref{l2}.

What remains is to prove (\ref{eq:final}). 
Substituting the result in Lemma~\ref{l3} below into  (\ref{eq:q}) and  performing the change of variable $\mtx{B^T} \fu =\fv$, we have 
\begin{align}
F_{\mtx{A},\mtx{B}}(\fp) &=  |\mtx{B}|^{-1} \int_{\mtx{A}(\fp-P)}  |\mtx{B}|^2 \, |h(\mtx{B^T} \fu)|^2 d\fu\nonumber \\
&=   \int_{\mtx{B^T} \mtx{A}(\fp-P)} |h(\fv)|^2 d\fv.
\label{parsi}
\end{align}
Since the boundary of $P$ has measure zero~\footnote{If the boundary has positive measure, one can obtain the same result using an approximation argument as in~\cite{szego}.}, we can assume that $\textbf{p}$ is an interior point of $P$, so that the set  $\fp-P$  contains a ball of non-zero measure centered at the origin. 
It then follows from Parseval's theorem that the integral (\ref{parsi}) converges to $(2 \pi)^{-N} m(Q)$ as $(\tau,\rho) \rightarrow (\tau_0, \rho_0)$, and the proof is complete.
\end{proof}
\begin{lem}
$h_{\mtx{B}} (\fu) = |\mtx{B}|\, h(\mtx{B^T} \fu)$
\label{l3}
\end{lem}
\begin{proof}[Proof of Lemma~\ref{l3}]
We have
\beq
\mathcal{F} h_{\mtx{B}}(\fy) = \mathds{1}_{\mtx{B} Q} (\fy) =  \mathds{1}_{Q}(\mtx{B}^{-1} \fy), 
\eeq
and the proof follows by computing the inverse transform
\begin{align}
h_{\mtx{B}}(\fu) &= \mathcal{F}^{-1} \mathds{1}_{\mtx{B} Q}(\fu) \nonumber \\
&= (2\pi)^{-N} \int_{\mathbb{R}^N} \mathds{1}_{\mtx{B} Q}(\fy) \exp(i \,  \fu \cdot \fy) d\fy \nonumber \\
&= (2\pi)^{-N} \int_{\mathbb{R}^N}\mathds{1}_{Q}(\mtx{B}^{-1} \fy) \exp(i \fu \cdot (\mtx{B} \mtx{B^{-1}}  \fy)) d\fy \nonumber \\
&= (2\pi)^{-N} |\mtx{B}|  \int_{\mathbb{R}^N}\mathds{1}_{Q}(\fz)  \exp(i \fu \cdot (\mtx{B} \fz)) d\fz \nonumber \\
&= (2\pi)^{-N} |\mtx{B}|  \int_{\mathbb{R}^N}\mathds{1}_{Q}(\fz)  \exp(i \fz \cdot (\mtx{B^T} \fu)) d\fz  \nonumber \\
&= |\mtx{B}| h(\mtx{B^T} \fu)
\end{align}
\end{proof}

\section*{Acknowledgment}
The author would like to thank Todd Kemp for navigating him through the subtleties of operator theory and Taehyung J. Lim for carefully reviewing the manuscript.
This work was partially supported by AFRL Award No. P07000236273 and Matrix Research inc. award No. FA8650-13-M-1556 through a subcontract with AFRL.

\end{document}